\documentclass[runningheads,a4paper]{llncs}

\usepackage{amssymb}%
\usepackage{graphicx}
\graphicspath{{figures/}}
\usepackage[disable]{todonotes}
\usepackage[hang,tight]{subfigure}
\usepackage[utf8]{inputenc} %
\usepackage{amsmath,amssymb} %
\usepackage{xspace}
\usepackage{makecell}

\usepackage{enumitem}
\usepackage{paralist}
\providecommand{\varitem}{} %


\usepackage{tikz}

\usepackage[english,linesnumbered,vlined,ruled,nokwfunc]{algorithm2e}
\DontPrintSemicolon
\makeatletter
\newcommand{\algorithmstyle}[1]{\renewcommand{\algocf@style}{#1}}
\makeatother
\SetKwComment{note}{$\triangleright$ }{}
\SetFuncSty{textsc}
\SetCommentSty{textit}
\SetDataSty{texttt}
\SetNlSty{text}{\color{black!100!white}}{}
\SetKwInOut{Initial}{Initial}
\SetKwInOut{Parameter}{Param.}
\SetKwInOut{Input}{Input}
\SetKwInOut{Data}{Data}
\SetKwInOut{Output}{Output}
\ResetInOut{Output} 
\let\oldnl\nl%
\newcommand{\nonl}{\renewcommand{\nl}{\let\nl\oldnl}}%
\SetKwProg{customProcedure}{Procedure}{}{}
\newcommand{\Procedure}[2]{\BlankLine\nonl\customProcedure{#1}{#2}}
\SetKw{KwAnd}{and} %

\def\BBPMCS{BBP-MCIS\xspace}
\def\TCMCS{$2$-MCIS\xspace}
\def\MWM{MWM\xspace}
\def\BC{\text{BC}\xspace}
\def\MCC{CC}
\def\S{\text{$\mathcal S$}\xspace}
\newcommand{\Part}[1]{\text{$\mathcal P_{#1}$}\xspace}

\newcommand{\Bl}{\text{Bl}\xspace}
\newcommand{\Br}{\text{Br}\xspace}

\usepackage{url}

\newcommand{\wIso}[0]{\ensuremath{\omega}\xspace}
\newcommand{\WIso}[0]{\ensuremath{\mathcal{W}}\xspace}

\DeclareMathOperator{\B}{B}
\DeclareMathOperator{\C}{C}

\DeclareMathOperator{\type}{type}
\DeclareMathOperator{\dom}{dom}

\begin{document}

\mainmatter  %

\title{Finding Largest Common Substructures of Molecules in Quadratic Time%
\thanks{This work was supported by the German Research Foundation (DFG), priority programme ``Algorithms for Big Data'' (SPP 1736).}}

\author{Andre Droschinsky \and Nils Kriege \and Petra Mutzel}

\authorrunning{A. Droschinsky \and N. Kriege \and P. Mutzel}

\institute{
Dept.\ of Computer Science,
Technische Universit{\"a}t Dortmund, Germany\\
\email{\{andre.droschinsky,nils.kriege,petra.mutzel\}@tu-dortmund.de}%
}

\maketitle

\begin{abstract}
Finding the common structural features of two molecules is a fundamental task
in cheminformatics. Most drugs are small molecules, which can naturally be 
interpreted as graphs. 
Hence, the task is formalized as maximum common subgraph problem. 
Albeit the vast majority of molecules yields outerplanar graphs this problem 
remains ${\sf NP}$-hard.

We consider a variation of the problem of high practical relevance, where
the rings of molecules must not be broken, i.e., the block and bridge structure 
of the input graphs must be retained by the common subgraph.
We present an algorithm for finding a maximum common connected induced subgraph 
of two given outerplanar graphs subject to this constraint. 
Our approach runs in time $\mathcal{O}(\Delta n^2)$ in outerplanar graphs on $n$
vertices with maximum degree $\Delta$. 
This leads to a quadratic time complexity in molecular graphs, which have bounded degree.
The experimental comparison on synthetic and real-world datasets 
shows that our approach is highly efficient in practice and outperforms 
comparable state-of-the-art algorithms.
\end{abstract}

\section{Introduction}
The maximum common subgraph problem arises in many application domains, where
it is necessary to elucidate common structural features of objects represented 
as graphs. In cheminformatics this problem has been extensively 
studied~\cite{Raymond2002,Ehrlich2011,Schietgat2013} and is often referred to as 
maximum or \emph{largest common substructure problem}.
Two variants of the problem can be distinguished: The maximum common induced 
subgraph problem (MCIS) is to find isomorphic induced subgraphs of two given 
graphs with the largest possible number of vertices. 
The maximum common edge subgraph problem (MCES) does not require that common 
subgraphs are induced and aims at maximizing the number of edges.
Both variants can be reduced to a maximum clique problem in the product graph
of the two input graphs~\cite{Raymond2002}.
In cheminformatics MCES is used more frequently since it (i) reflects
the notion of chemical similarity more adequately~\cite{Raymond2002}, and (ii)
can reduce the running time of product graph based algorithms~\cite{Nicholson1987}.
Although such algorithms still have exponential running time in the worst case,
they are commonly applied to molecular graphs in practice~\cite{Raymond2002}.

However, there are several restricted graph classes which render polynomial time 
algorithms possible~\cite{Akutsu1993,Akutsu2013,Yamaguchi2004}.
The seminal work in this direction is attributed to J.~Edmonds~\cite{Matula1978}, 
who proposed a polynomial time algorithm for the maximum common subtree problem.
Here, the given graphs and the desired common subgraph must be trees. 
Recently, it was shown that this problem can be solved in time 
$\mathcal{O}(\Delta n^2)$ for (unrooted) trees on $n$ vertices with maximum 
degree $\Delta$~\cite{Droschinsky2016}.
The (induced) subgraph isomorphism problem (SI) is to decide if a pattern graph
is isomorphic to an (induced) subgraph of another graph and is generalized by 
MCIS and MCES, respectively.
Both variants of SI are ${\sf NP}$-complete, even when the pattern is a 
forest and the other graph a tree~\cite{Garey1979}; just as when the pattern is 
a tree and the other is outerplanar~\cite{Syslo1982}. 
On the other hand, when both graphs are biconnected and outerplanar, 
induced SI can be solved in time $\mathcal{O}(n^2)$~\cite{Syslo1982} and SI in 
$\mathcal{O}(n^3)$~\cite{Lingas1989}.
These complexity results and the demand in cheminformatics lead to the 
consideration of MCES under the so-called \emph{block and bridge preserving} 
(BBP) constraint~\cite{Schietgat2013}, which requires the common subgraph to retain 
the local connectivity of the input graphs. BBP-MCES is not only computable
in polynomial-time, but also yields meaningful results for cheminformatics. 
A polynomial-time algorithm was recently proposed for BBP-MCIS, which requires 
time $\mathcal{O}(n^6)$ in series-parallel and $\mathcal{O}(n^5)$ in outerplanar 
graphs~\cite{Kriege2014a}.

Most of the above mentioned polynomial time algorithms are either not applicable 
to molecular graphs or impractical due to high constants. A positive exception is 
the BBP-MCES approach of~\cite{Schietgat2013}, which has been shown to outperform 
state-of-the-art algorithms on molecular graphs in practice.
This algorithm is stated to have a running time of $\mathcal{O}(n^{2.5})$, but 
in fact leads to a running time of $\Omega(n^4)$ in the worst case~\cite{Droschinsky2016}.

\subsubsection{Our contribution.}
We take up the concept of BBP and propose a novel BBP-MCIS algorithm with running 
time $\mathcal{O}(\Delta n^2)$ in outerplanar graphs with $n$ vertices and maximum 
degree $\Delta$. 
We obtain this result by combining ideas of \cite{Droschinsky2016} for the 
maximum common subtree problem with a new algorithm for biconnected MCIS in 
biconnected outerplanar graphs.
For this subproblem we develop a quadratic time algorithm, which exploits the 
fact that the outerplanar embedding of a biconnected outerplanar graph is unique.
Moreover, the algorithm allows to list all solutions in quadratic total time.
Our approach supports to solve BBP-MCIS w.r.t.\@ a weight function on the 
mapped vertices and edges.
The experiments show that BBP-MCIS in almost all cases yields the same 
results as BBP-MCES for molecular graphs under an adequate weight function.
Our method outperforms in terms of efficiency the BBP-MCES approach of~\cite{Schietgat2013} 
by orders magnitude.

\section{Preliminaries}
\label{sec:Preliminaries}
We consider simple undirected graphs.
Let $G=(V,E)$ be a graph, we refer to the set of \emph{vertices} $V$ by $V(G)$ 
or $V_G$ and to the set of \emph{edges} by $E(G)$ or $E_G$. An edge connecting 
two vertices $u, v \in V$ is denoted by $uv$ or $vu$. The \emph{order} $|G|$ of 
a graph $G$ is its number of vertices.
Let $V'\subseteq V$, then the graph $G[V']=(V',E')$ with 
$E'=\{uv\in E \mid u,v\in V'\}$ is called \emph{induced} subgraph. For $U \subseteq V$ 
we write $G\setminus U$ for $G[V\setminus U]$. A graph is \emph{connected} 
if there is a path between any two vertices. A \emph{connected component} of a 
graph $G$ is a maximal connected subgraph of $G$. A graph $G=(V,E)$ with $|V| \geq 3$
is called \emph{biconnected} if $G\setminus\{v\}$ is connected for any $v \in V$.
A maximal biconnected subgraph of a graph $G$ is called \emph{block}. If an edge 
$uv$ is not contained in any block, the subgraph $(\{u,v\},\{uv\})$ is called a 
\emph{bridge}. A vertex $v$ of $G$ is called \emph{cutvertex}, if $G\setminus \{v\}$
consists of more connected components than $G$. 
A graph is \emph{planar} if it admits a drawing on the plane such that no two 
edges cross. The connected regions of the drawing enclosed by the edges are 
called \emph{faces}, the unbounded region is referred to as \emph{outer face}.
An edge and a face are said to be \emph{incident} if the edge touches the face.
Two faces are \emph{adjacent} if they are incident with a common edge.
A graph is called \emph{outerplanar} if it admits a drawing on the plane without 
crossings, in which every vertex lies on the boundary of the outer face.
A \emph{matching} in a graph $G=(V,E)$ is a set of edges $M\subseteq E$, such 
that no two edges share a vertex. A matching $M$ is \emph{maximal} if there is 
no other matching $M'\supsetneq M$ and \emph{perfect}, if $2|M|=|V|$.
A \emph{weighted graph} is a graph endowed with a function $w:E\to\mathbb R$. 
A matching $M$ in a weighted graph has weight by $W(M):=\sum_{e\in M} w(e)$;  
it is a \emph{maximum weight matching} (\MWM) if there is no matching $M'$ of $G$ 
with $W(M')>W(M)$.

An \emph{isomorphism} between two graphs $G$ and $H$ is a bijection 
$\phi : V(G) \to V(H)$ such that $uv\in E(G)\Leftrightarrow\phi(u)\phi(v)\in E(H)$.
A \emph{common (induced) subgraph isomorphism} is an isomorphism between (induced)
subgraphs $G' \subseteq G$ and $H' \subseteq H$.
A subgraph $G' \subseteq G$ is \emph{block and bridge preserving} (BBP) if
(i) each bridge in $G'$ is a bridge in $G$,
(ii) any two edges in different blocks in $G'$ are in different blocks in $G$.
A common subgraph isomorphism $\phi$ is BBP if both subgraphs are BBP, it is
\emph{maximal} if it cannot be extended.
Molecular graphs are typically annotated with atom and bond types, which should
be preserved under isomorphisms. More general, we allow for a weight function 
$\wIso:(V_G\times V_H)\cup (E_G\times E_H)\to \mathbb{R}^{\geq 0}\cup \{-\infty\}$.
The weight $\WIso(\phi)$ of an isomorphism $\phi$ between $G$ and $H$ under 
$\wIso$ is the sum of the weights $\wIso(v,\phi(v))$ and $\wIso(uv,\phi(v)\phi(v))$
for all vertices $v$ and edges $uv$ mapped by $\phi$.
A common subgraph isomorphism $\phi$ is \emph{maximum} if its weight $\WIso(\phi)$ 
is maximum.
A maximum isomorphism does not map any vertices or edges contributing weight 
$-\infty$ and we call these pairs \emph{forbidden}.
We further define $[1..k]:=\{1,\ldots,k\}$ for $k\in \mathbb N$.

\section{Biconnected MCIS in Outerplanar Graphs}
\label{sec:2MCS}
In this section we present an algorithm to determine the weight of a maximum 
common biconnected induced subgraph isomorphism (\TCMCS) of two biconnected 
outerplanar graphs.
First we show how to compute the maximal common biconnected subgraph isomorphisms.
Since these may contain forbidden vertex and edge pairs, we then describe how to 
obtain the weight of a maximum solution from them. Finally we show how to output 
one or all maximum solutions.

Outerplanar graphs are well-studied and have several characteristic properties, 
see \cite{Syslo1982} for further information. In particular, our algorithm
exploits the fact that biconnected outerplanar graphs have a unique outerplanar
embedding in the plane (up to the mirror image). In these embeddings, every edge 
is incident to exactly two faces that are uniquely defined.
We observe that the mapping is determined by starting parameters, i.e., 
an edge of both input graphs together with the mapping of their endpoints and 
incident faces.

We say a face is mapped by an isomorphism $\phi$ if all the vertices bordering 
the face are mapped by $\phi$.
We distinguish four cases to describe the mapping of an edge $uv \in E(G)$ to an
edge $u'v' \in E(H)$ by an isomorphism $\phi$ between biconnected induced 
subgraphs.
Assume the edge $uv$ is incident to the faces $A$ and $B$ in $G$ and $u'v'$ is 
incident to $A'$ and $B'$ in $H$, see Fig.~\ref{fig:faces}.
At least one face incident to $uv$ must be mapped by $\phi$, since the common 
subgraph must be biconnected. For the sake of simplicity of the case distinction, 
we also associate the two other faces, regardless of whether they are mapped or not.
The isomorphism may map the endpoints of the edges in two different ways---just 
as the two incident faces. 
We can distinguish the following four cases:\\
(1) $u \mapsto u',\ v \mapsto v',\ A \mapsto A',\ B \mapsto  B'$, \hfill
(2) $u \mapsto v',\ v \mapsto u',\ A \mapsto A',\ B \mapsto  B'$,\\
(3) $u \mapsto u',\ v \mapsto v',\ A \mapsto B',\ B \mapsto  A'$, \hfill 
(4) $u \mapsto v',\ v \mapsto u',\ A \mapsto B',\ B \mapsto  A'$.\\
Given an isomorphism $\phi$ between biconnected common induced subgraphs that 
maps the two endpoints of an edge $e$, let the function $\type(e,\phi) \in [1..4]$
determine the type of the mapping as above.
The following result is the key to obtain our efficient algorithm.

\begin{lemma}\label{lem:outer:determinism}
 Let $\phi$ and $\phi'$ be maximal isomorphisms between biconnected common 
 induced subgraphs of the biconnected outerplanar graphs $G$ and $H$. 
 Assume $e \in E(G)$ is mapped to the same edge $e' \in E(H)$ by $\phi$ and 
 $\phi'$, then
 \begin{equation*}
   \type(e,\phi) = \type(e,\phi') \Longleftrightarrow \phi' = \phi.
 \end{equation*}
\end{lemma}
\begin{proof}
 It is obvious that the direction $\Longleftarrow$ is correct. 
 We prove the implication $\Longrightarrow$. 
 Since the common subgraph is required to be biconnected, 
 the isomorphisms $\phi$ and $\phi'$ both must map at least one face of 
 $G$ incident to the edge $e$ to a face of $H$ incident to $e'$. 
 The two faces as well as the mapping of endpoints of the two edges are uniquely 
 determined by the type of the mapping.
 We consider the mapping of the vertices on the cyclic border of these faces. 
 Since the mapping of the endpoints of $e$ are fixed, the mapping of 
 all vertices on the border of the face is unambiguously determined.
 Since the common subgraph is required to be biconnected, every extension of 
 the mapping must include all the vertices of a neighboring face.
 For this face, again, the mapping of the endpoints of the shared edge implicates
 the mapping of all vertices on the cyclic border and the extension is unambiguous.
 Therefore, the mapping can be successively extended to an unmapped face.
 Consequently $\phi(u) = \phi'(u)$ holds for all $u \in \dom(\phi) \cap \dom(\phi')$.
 Since $\phi$ and $\phi'$ are  maximal it is not possible that one of them can be
 extended and, hence, we must have $\dom(\phi) = \dom(\phi')$ and the result follows.
 \qed
\end{proof}

\SetKwFunction{PART}{SplitIso}
\SetKwFunction{MAXMATCH}{MaximalIso}
The proof of Lemma~\ref{lem:outer:determinism} constructively shows how to 
obtain a maximal solution given two edges $uv \in E(G)$, $u'v' \in E(H)$ and a 
type parameter $t \in [1..4]$. We assume that this approach is realized by the 
procedure \MAXMATCH{$uv, u'v', t$}, which returns the unique maximal isomorphism 
that maps the two given edges according to the specified type. The algorithm can 
be implemented by means of a tree structure that encodes the neighboring relation 
between inner faces, e.g., SP-trees as in~\cite{Kriege2014a,Kriege2014b} or
weak dual graphs similar to the approach of \cite{Syslo1982}.
The running time to compute a maximal solution $\phi$ then is 
$\mathcal{O}(|\phi|) \subseteq \mathcal{O}(n)$.
Note that for some edge pairs not all four types of mappings are possible. 
The type $t \in [1..4]$ is \emph{valid} for a pair of edges if at least 
one incident face can be mapped according to type $t$, i.e., the edges are 
incident to faces that are bordered by the same number of vertices. 

A maximal solution $\phi$ may map vertex and edge pairs that are forbidden
according to the weight function. 
In order to obtain the maximum weight, we split $\phi$ into 
\emph{split isomorphisms} $\phi_1,\dots,\phi_k$ such that each 
(i) has non-negative weight and 
(ii) again is an isomorphism between biconnected induced common subgraphs. 
The split isomorphisms can be obtained in time $\mathcal{O}(|\phi|)$ as follows.
We consider the graph $G' = G[\dom(\phi)]$.
For every forbidden edge $uv$ that is incident to two inner faces in $G'$, we split 
the graph into $G'_i[V(C_i)\cup\{u,v\}]$, where $C_i$ is a connected component of
$G'\setminus\{u,v\}$, $i \in [1..2]$. In these graphs we delete the forbidden 
vertices and edges and determine the blocks $B_1,\dots,B_k$. 
Then $\phi$, restricted to the vertices $V(B_i)$ of a block $B_i$,
yields the split isomorphism $\phi_i$ for $i \in [1..k]$. This approach is 
realized by the function \PART{$\phi$} used in the following.
Every edge $e \in E(G)$ is mapped by at most one of the resulting isomorphisms,
referred to by $\phi_e$.
Every \TCMCS is a split isomorphism obtained from some maximal solution.

\begin{algorithm}[tb]
  \caption{\TCMCS in outerplanar graphs}
  \label{alg:2mcs:outer}
  \Input{Biconnected outerplanar graphs $G$ and $H$.}
  \Output{Weight of a maximum common biconnected subgraph isomorphism.}
  \Data{Table $D(e,f,t)$, $e \in E(G)$, $f \in E(H)$, $t \in [1..4]$ 
        storing the weight of a \TCMCS $\phi$ mapping $e$ to $f$ with $\type(e,\phi)=t$.
  }

  \ForAll{$uv \in E(G)$, $u'v' \in E(H)$ and $t \in [1..4]$} {
    \If{type $t$ valid for $uv$ and $u'v'$ \KwAnd $D(uv,u'v',t)$ undefined} {\label{alg:2mcs:outer:table}
      $\phi \gets \MAXMATCH(uv,u'v',t)$ \;
      $(\phi_1, \dots, \phi_k) \gets \PART(\phi)$ \;\label{alg:2mcs:outer:split}
      \ForAll{edges $e \in E(G)$ mapped to $f \in E(H)$ by $\phi$} {
         $D(e,f,\type(e,\phi)) \gets \left\{\begin{array}{ll} 
             W(\phi_e) & \ \text{if $e$ is mapped by the split iso. $\phi_e$} \\
             -\infty   & \ \text{otherwise.}\end{array}\right.$\label{alg:2mcs:outer:weight}\;
      }
    }
  }
  \Return maximum entry in $D$ \label{alg:2mcs:outer:return}

\end{algorithm}

Algorithm~\ref{alg:2mcs:outer} uses a table $D(e,f,t)$, $e \in E(G)$, 
$f \in E(H)$, $t \in [1..4]$ storing the weight of a \TCMCS under the
constraint that it maps $e$ to $f$ according to type $t$.
The size of the table is $4 |E(G)| |E(H)| \in \mathcal{O}(nm)$, where $n=|V(G)|$
and $m=|V(H)|$.
The algorithm starts with all pairs of edges and all valid types of mappings 
between them. For each, the maximal isomorphism between biconnected common induced
subgraphs is computed by extending this initial mapping.
By splitting the maximal solution, multiple valid isomorphisms with non-negative
weight are obtained.
These weights are then stored in $D$ for all pairs of edges contained in $\phi$ 
considering the type of the mapping.
This includes the $-\infty$ weights occurring if there are forbidden vertices or edges.
Keeping these values allows to avoid generating the same isomorphism multiple times.
The main procedure loops over all pairs of edges and the four possible mappings
for each pair. Note that a mapping $\phi$ and its split isomorphisms are computed
in time $\mathcal{O}(|\phi|) \subseteq \mathcal{O}(n)$. Improved analysis gives the 
following result.

\begin{theorem}
\label{2mcs:proof}
 Algorithm~\ref{alg:2mcs:outer} computes the weight of a \TCMCS between 
 biconnected outerplanar graphs $G$ and $H$ in time $\mathcal{O}(|G||H|)$.
\end{theorem}
\begin{proof}
 We allocate the costs for a call of \MAXMATCH followed by \PART to cells of the
 table $D$. A mapping $\phi$ containing $k$ edges is computed in time 
 $\mathcal{O}(k)$ and as a result exactly $k$ cells of the table $D$ are filled 
 with a value. 
 The value of a cell is computed at most once: Line~\ref{alg:2mcs:outer:table}
 assures that an edge mapping of a specific type is not used as initial mapping
 when the corresponding cell is already filled.
 Every initial mapping that is extended must lead to an isomorphism containing 
 only edge mappings associated with undefined cells according to 
 Lemma~\ref{lem:outer:determinism}.
 Therefore the total costs of the algorithm can be allocated to cells of $D$, 
 such that each cell pays a constant amount.
 This proves that the total running time is bounded by the size of the table, 
 which is $\mathcal{O}(|G||H|)$. \qed
\end{proof}

We can easily modify the algorithm to enumerate all maximum isomorphisms without 
affecting the total running time. 
First we run Algorithm~\ref{alg:2mcs:outer} once to obtain the maximum weight 
$W_\text{max}$. Then we run a modified version of Algorithm~\ref{alg:2mcs:outer} 
that outputs every split isomorphism $\phi_i$ of size $W(\phi_i)=W_\text{max}$ 
as soon as it is found, right after  $\PART(\phi)$ is called in line \ref{alg:2mcs:outer:split}.

\section{Solving \BBPMCS in Outerplanar Graphs}
\label{sec:CompSingleMCS}
In the previous section we have presented an algorithm to compute a \TCMCS between two biconnected outerplanar graphs.
In this section we will generalize it to compute a \BBPMCS between two outerplanar graphs $G$ and $H$.
In the following we assume the isomorphisms to be BBP.
We require the input graphs to be connected.
Otherwise we compute a \BBPMCS for all pairs of connected components and select an isomorphism of maximum weight. 

We proceed as follows.
First, we give insight into the \emph{BC-tree} data structure, which helps to partition the set $\S$ of all BBP common subgraph isomorphisms between $G$ and $H$ into subsets w.r.t.\,certain conditions.
Then we compute an isomorphism of maximum weight in each of the subsets using a dynamic programming approach similar to the one used in~\cite{Droschinsky2016} to solve the maximum common subtree problem.
Among the computed isomorphisms we output one with maximum weight, thus a \BBPMCS.

\subsubsection{The BC-tree data structure.}
Given a \BBPMCS, we can observe that bridges of $G$ are mapped to bridges of $H$ and that edges in one block of $G$ can only be mapped to edges contained in exactly one block of $H$, such that the mapped edges form a biconnected common subgraph. 
For a connected graph $G$ let $\C^G$ denote the set of cutvertices, $\Bl^G$ the 
set of blocks and $\Br^G$ the set of bridges and $\B^G:=\Bl^G\cup \Br^G$. 
The \emph{BC-tree} ${\BC^G}$ of $G$ is the tree with nodes $\B^G \cup \C^G$ and 
edges between nodes $b \in \B^G$ and $c\in \C^G$ iff $c\in V(b)$.
We refer to the vertices of the BC-tree as B- and C-nodes and distinguish block 
nodes from bridge nodes.
An example of a graph $G$ and its BC-tree $\BC^G$ is shown in Fig.~\ref{fig:bc}.
For any graph $G$, we define  \MCC$(V',U)$ as the connected component of $G[V']$ that includes at least one vertex of $U$.
We allow only such sets $U$, where the component is unambiguous.
For example, in Fig.~\ref{fig:bc}, \MCC$(V_G\setminus V_{b_2},V_{b_4})$ is the graph $G[\{c_3,u,v\}]$.

\tikzstyle{Vertex}=[circle, draw = black, fill=black, inner sep=0pt, minimum width=4.2pt]
\tikzstyle{Edge} = [draw,thick,-]
\tikzstyle{BlNode}=[rounded corners=2, thick, draw=black, fill=lightgray, minimum size = 1.0em]
\tikzstyle{BrNode}=[rounded corners=2, thick, draw=black, minimum size = 1.0em]

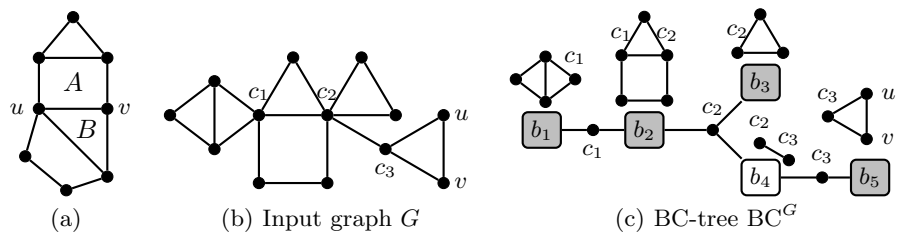
\begin{figure}[t]
	\centering
	\subfigure[]{
		\begin{tikzpicture}[scale = 0.90]
			\node (d) at (2,1) [Vertex,label=left:\small{$u$}] {};
			\node (e) at (2,1.75) [Vertex] {}
				edge [Edge] (d);
			\node (c) at (3,1.75) [Vertex] {}
				edge [Edge] (e);
			\node (f) at (1.8,0.3) [Vertex] {}
				edge [Edge] (d);
			\node (g) at (3,0) [Vertex] {}
				edge [Edge] (d);
			\node (i) at (2.4,-0.2) [Vertex] {}
				edge [Edge] (f)
				edge [Edge] (g);
			\node (h) at (3,1) [Vertex,label=right:\small{$\!v$}] {}
				edge [Edge] (d)
				edge [Edge] (c)
				edge [Edge] (g);
			\node (b) at (2.5,2.35) [Vertex] {}
				edge [Edge] (e)
				edge [Edge] (c);
			\node (A) at (2.5,1.4) {\normalsize{$A$}};
			\node (B) at (2.7,0.7) {\normalsize{$B$}};
		\end{tikzpicture}
		\label{fig:faces}
	}
	\subfigure[Input graph $G$]{
		\begin{tikzpicture}[scale = 0.90]
			\node (a) at (0.7,1) [Vertex] {};
			\node (b) at (1.35,1.5) [Vertex] {}
				edge [Edge] (a);
			\node (c) at (1.35,0.5) [Vertex] {}
				edge [Edge] (a)
				edge [Edge] (b);
			\node (d) at (2,1) [Vertex,label=above:\small{$c_1$}] {}
				edge [Edge] (b)
				edge [Edge] (c);
			\node (e) at (2.5,1.85) [Vertex] {}
				edge [Edge] (d);
			\node (f) at (2,0) [Vertex] {}
				edge [Edge] (d);
			\node (g) at (3,0) [Vertex] {}
				edge [Edge] (f);
			\node (h) at (3,1) [Vertex,label=above:\small{$c_2$}] {}
				edge [Edge] (d)
				edge [Edge] (e)
				edge [Edge] (g);
			\node (i) at (3.5,1.85) [Vertex] {}
				edge [Edge] (h);
			\node (j) at (4,1) [Vertex] {}
				edge [Edge] (i)
				edge [Edge] (h);
			\node (k) at (3.85,0.5) [Vertex,label=below:\small{$c_3$}] {}
				edge [Edge] (h);
			\node (l) at (4.7,1) [Vertex,label=right:\!\small{$u$}] {}
				edge [Edge] (k);
			\node (m) at (4.7,0) [Vertex,label=right:\!\small{$v$}] {}
				edge [Edge] (l)
				edge [Edge] (k);
		\end{tikzpicture}
		\label{fig:bc:graph}
	}
	\subfigure[BC-tree $\BC^G$]{	
		\begin{tikzpicture}[scale = 0.90]
			\node (b1) at (0.5,0) [BlNode] {$b_1$};
			\node (c1) at (1.25,0) [Vertex,label=below:\small{$c_1$}] {}
				edge [Edge] (b1);
			\node (b2) at (2,0) [BlNode] {$b_2$}
				edge [Edge] (c1);
			\node (c2) at (3,0) [Vertex,label=above:\small{$c_2$}] {}
				edge [Edge] (b2);
			\node (b3) at (3.7,0.7) [BlNode] {$b_3$}
				edge [Edge] (c2);
			\node (b4) at (3.7,-0.7) [BrNode] {$b_4$}
				edge [Edge] (c2);
			\node (c3) at (4.6,-0.7) [Vertex,label=above:\small{$c_3$}] {}
				edge [Edge] (b4);
			\node (b5) at (5.3,-0.7) [BlNode] {$b_5$}
				edge [Edge] (c3);
			\node at (0.65,0.8) {
				\begin{tikzpicture}[scale = 0.6]
					\node (a) at (0.7,1) [Vertex] {};
					\node (b) at (1.35,1.5) [Vertex] {}
						edge [Edge] (a);
					\node (c) at (1.35,0.5) [Vertex] {}
						edge [Edge] (a)
						edge [Edge] (b);
					\node (d) at (2,1) [Vertex,label=above:\small{$c_1$}] {}
						edge [Edge] (b)
						edge [Edge] (c);
				\end{tikzpicture} };
			\node at (2,1.05) {
				\begin{tikzpicture}[scale = 0.6]
					\node (d) at (2,1) [Vertex,label=above:\small{$c_1$}] {};
					\node (e) at (2.5,1.85) [Vertex] {}
						edge [Edge] (d);
					\node (f) at (2,0) [Vertex] {}
						edge [Edge] (d);
					\node (g) at (3,0) [Vertex] {}
						edge [Edge] (f);
					\node (h) at (3,1) [Vertex,label=above:\small{$c_2$}] {}
						edge [Edge] (d)
						edge [Edge] (e)
						edge [Edge] (g);
				\end{tikzpicture} };
			\node at (3.6,1.42) {
				\begin{tikzpicture}[scale = 0.6]
					\node (h) at (3,1) [Vertex,label=above:\small{$c_2$}] {};
					\node (i) at (3.5,1.85) [Vertex] {}
						edge [Edge] (h);
					\node (j) at (4,1) [Vertex] {}
						edge [Edge] (i)
						edge [Edge] (h);
				\end{tikzpicture} };
			\node at (3.9,-0.1) {
				\begin{tikzpicture}[scale = 0.45]
					\node (h) at (3,1) [Vertex,label=above:\small{$c_2$}] {};
					\node (k) at (3.85,0.5) [Vertex,label=above:\small{$c_3$}] {}
						edge [Edge] (h);
				\end{tikzpicture} };
			\node at (5.1,0.2) {
				\begin{tikzpicture}[scale = 0.6]
					\node (k) at (3.85,0.5) [Vertex,label=above:\small{$c_3$}] {};
					\node (l) at (4.7,1) [Vertex,label=right:\small{$u$}] {}
						edge [Edge] (k);
					\node (m) at (4.7,0) [Vertex,label=right:\small{$v$}] {}
						edge [Edge] (l)
						edge [Edge] (k);
				\end{tikzpicture} };
		\end{tikzpicture}
		\label{fig:bc:graph_bc-tree}
	}
	\caption{A biconnected outerplanar graph~\subref{fig:faces} with an edge $uv$ incident to the faces $A$ and $B$; a connected outerplanar graph \subref{fig:bc:graph} and its BC-tree \subref{fig:bc:graph_bc-tree}. Block nodes have a gray background, while bridge nodes are not filled. The solid black nodes are the cutvertices. The corresponding subgraphs of $G$ are shown above the block and bridge nodes.}
	\label{fig:bc}
\end{figure}

\subsubsection{Partitioning of all BBP isomorphisms $\S$ into $\S=\bigcup_x \S_x$.}
First, we define $\S_1$ and $\S_2$.
Let $b\in \B^G$ be an arbitrary block or bridge in $G$.
We define $\S_1$ to contain all isomorphisms $\phi$ where at least one edge in $b$ is mapped by the isomorphisms, i.e., $|\dom(\phi)\cap V(b)|\geq 2$.
$\S_2$ is defined to contain all isomorphisms where exactly one vertex in $b$ is mapped by the isomorphism.
We can observe that $S_1$ and $S_2$ are disjoint and all other isomorphisms between $G$ and $H$ do not contain any vertices of $b$.
Let $N=\{b_1,\ldots, b_k\} \subseteq \B^G$ be the blocks and bridges that share a 
cutvertex with $b$, i.e., $b_i \in N$ iff there is a node $c\in \C^G$ with $bc$ and $cb_i$ 
edges in the BC-tree $\BC^G$.
Any isomorphism $\phi$ that maps no vertex of $b$, maps vertices of at most one node $b_i$, because $G[\dom(\phi)]$ is connected by definition.
For every $b_i$ we recursively define sets $\S_x$ of isomorphisms as described above that map only vertices of \MCC$(V_G\setminus V_b,V_{b_i})$.

As example consider Fig.~\ref{fig:bc:graph_bc-tree} and let $b:=b_2$.
$\S_1$ consist of isomorphisms which map at least one edge of $b_2$ to an edge in $H$. The isomorphisms in $\S_2$ map exactly one vertex of $V(b)$ to $H$.
The recursion continues on $N=\{b_1,b_3,b_4\}$.
Three additional sets consist of isomorphisms which map at least one edge (and three more for exactly one vertex) of $V(b_i),\ i\in\{1,3,4\}$, but no vertex of $V(b_2)$, operating on \MCC$(V_G\setminus V_{b_2},V_{b_i})$.
The recursion for $b:=b_4$ continues with $N=\{b_5\}$ and two additional sets.
Some of the sets $\S_x$ are empty.

\subsubsection{Partitioning of $\S_x$ into $\S_x=\bigcup_y \Part{xy}$.}
Before computing an isomorphism of maximum weight in a set $\S_x$, we partition $\S_x$ into subsets $\Part{x1},\Part{x2},\ldots$.
The focus for the separation now is on the graph $H$.
We distinguish two cases.
If $\S_x$ is a set, where at least one edge of a certain block (bridge) $b$ is mapped, then $\S_x$ is partitioned into $|\Bl^H|\ (|\Br^H|)$ subsets.
The meaning is that for each B-node $\bar b\in Bl^H\ (\bar b \in \Br^H)$ the mapped vertices of the B-node $b\in B^G$ are mapped only to $V(\bar b)$.
In terms of BBP this is block (bridge) preserving between $b$ and $\bar b$, as intended.
If $\S_x$ is a set, where exactly one vertex of $b$ is mapped, the subsets are defined as follows.
For each $(v,\bar v)\in V(b)\times V(H)$, where $\wIso(v\bar v)\neq -\infty$ and $v$ is in the \MCC{} we operate on, we define a subset with the restriction $\phi(v)=\bar v$.

\subsubsection{Computing a maximum isomorphism in a subset $\Part{xy}$.}
We now describe how to compute an isomorphism $\phi$ of maximum weight in a subset $\Part{xy}\subseteq\S_x$.
The idea is to recursively extend mappings between some vertices of two single bridges or two single blocks along all pairs of mapped cutvertices into other B-nodes determined by \MWM{}s, while preserving bridges and blocks.
Between the computed isomorphisms we select one of maximum weight.

First, let $\Part{xy}$ be a subset, where at least one edge of a B-node $b\in \B^G$ has to be mapped to an edge of a B-node $\bar b\in \B^H$.
If $b$ and $\bar b$ are bridges, the two possible mappings $V(b)\to V(\bar b)$ are considered.
If both are blocks, all maximal common biconnected subgraph isomorphisms between the blocks are considered (cf.\,Alg.~\ref{alg:2mcs:outer}).
We may have given a fixed mapping $v\mapsto \bar v$ (cf.\,(i) below).
We call a considered isomorphism valid, if it respects the possible fixed mapping and contains only vertices of the \MCC{} we are operating on.
We extend all the valid isomorphisms $\phi$ along all pairs $\phi(c)=\bar c, c\neq v$ of mapped cutvertices as follows.
Let $B_c:=\{b_1,\ldots b_k\}$, be the B-nodes of $\B^G$, where $bcb_i$ is a path, and $\bar B_c:=\{\bar b_1,\ldots \bar b_l\}$, be the B-nodes of $\B^H$, where $\bar b \bar c \bar b_j$ is a path, $i\in[1..k],j\in[1..l]$. 
For each pair $(b_i,\bar b_j)\in B_c\times \bar B_c$ we recursively calculate a \BBPMCS $\varphi_{ij}$ under the following restrictions: 
(i) The cutvertices must be mapped: $c\mapsto \bar c$. 
(ii) $b_i$ and $\bar b_j$ are both bridges or both blocks. 
(iii) At least one other vertex in the block (bridge) $b_i$ must be mapped, but only to $V(\bar b_j)$. 
Restriction (iii) assures that at least one vertex is added to the isomorphism. 
Therefore, the recursion to compute $\varphi_{ij}$ is the method described in this paragraph.
After computing $\varphi_{ij}$ for each pair $(b_i,\bar b_j)$, we construct a weighted bipartite graph with vertices $B_c\uplus\bar B_c$ for each pair of mapped cutvertices. 
The weight of each edge $b_i\bar b_j$ is determined by the weight of a \BBPMCS under the above restrictions, subtracted by $\wIso(c,\bar c)$ for the appropriate cutvertices $c$ and $\bar c$.
If there in no such restricted \BBPMCS, there is no edge.
Computing a \MWM on each of the bipartite graphs determines the extension of $\phi$. 
For each matching edge the corresponding computed isomorphisms are merged with $\phi$. 
After extending all valid isomorphisms, we select one of maximum weight.%

Second, let $\Part{xy}$ be a subset, where exactly one vertex $v$ of $V(b)$ is mapped, and let $\phi(v)=\bar v$. 
If $v$ is no cutvertex, the only possible expansion is within $V(b)$, which is not allowed in this subset.
Therefore this subset contains exactly one isomorphism, $v\mapsto\bar v$. 
Next, assume $v$ is a cutvertex. If $\bar v$ is a cutvertex, we may extend $\phi$ similar to the previous paragraph. 
In doing so, $c:=v$, $\bar c:=\bar v$ and $B_c$ as before.
The only difference is $\bar B_c$, which is defined by all B-nodes containing $\bar v=\bar c$.
The reason is that we have not mapped any other vertices yet, therefore we may expand in all directions in $H$.
If $\bar v$ is no cutvertex, then $\bar v$ is contained in exactly one $\bar b \in B^H$. 
We are interested in BBP isomorphisms only. This means, all vertices that are mapped to $V(\bar b)$ must be in the same block or bridge $b'\in \B^G$. 
Therefore, for each $b'\in \B^G$, where $bvb'$ is a path and $b'$ and $\bar b$ are of the same type (bridge/block), we compute an isomorphism with fixed mapping $v\mapsto \bar v$, where at least one edge of $b'$ is mapped to $\bar b$.
This falls back to the method of the above paragraph as well. 
Among the computed isomorphisms we select one of maximum weight. 
The appendix lists the pseudocode for computing a \BBPMCS as described above.

\subsubsection{Time Complexity.}
The time to compute a \BBPMCS essentially depends on the time to compute the BC-trees, the biconnected  isomorphisms between the  blocks of $G$ and $H$, and the time to compute all the \MWM{}s.
The time to compute a BC-tree is linear in the number of edges and vertices. \todo{Quelle?}
Considering all pairs of blocks and Theorem~\ref{2mcs:proof} we can bound the time for computing all the biconnected isomorphisms by $\mathcal{O}(\sum_b\sum_{\bar{b}} |V_b||V_{\bar b}|)\subseteq \mathcal{O}(|G||H|)$.
We only need to compute \MWM{}s for the pairs of cutvertices of the two graphs.
It follows from the result of \cite[Theorem 7]{Droschinsky2016} for the maximum 
common subtree problem, that the total time for this is 
$\mathcal{O}(|G||H|(\min \{\Delta^G,\Delta^H\}+\log\max \{\Delta^G,\Delta^H\}))$,
where $\Delta^{\cal G}$ is the maximum degree of a C-node in $\BC^\mathcal{G}$.
This proves the following theorem.

\begin{theorem}
\label{th:runtime:bbpmcs}
\BBPMCS between two outerplanar graphs $G$ and $H$ can be solved in time $\mathcal{O}(|G||H| \Delta(G,H))$, where $\Delta(G,H)=1$ iff $G$ or $H$ is biconnected or both are of bounded degree; otherwise $\Delta(G,H)=\min \{\Delta^G,\Delta^H\}+\log\max \{\Delta^G,\Delta^H\}$.
\end{theorem}

\section{Experimental Evaluation}
In this section we evaluate our BBP-MCIS algorithm experimentally and compare 
to the BBP-MCES approach of~\cite{Schietgat2013}.\footnote{We are grateful to 
Leander Schietgat for providing the implementation used in~\cite{Schietgat2013}.}
Both algorithms were implemented in C\texttt{++} and compiled with GCC v.4.8.4 
as 64-bit application.
Running times were measured on an Intel Core i7-3770 CPU using a single core.
The available memory of 16 GB was sufficient for all the computations.

We are interested in answering the following questions:
\begin{description}
\item[(H1)] To what extent differs \BBPMCS from BBP-MCES on molecular graphs?
\item[(H2)] How large is the difference in terms of running time on molecular graphs?
\item[(H3)] How is the running time affected by specific properties of the input graphs?
\end{description}

To answer {\bf(H1)} and {\bf(H2)} we extracted 29000 randomly chosen pairs of outerplanar molecular graphs from a large chemical database.\footnote{NCI Open Database, GI50, \url{http://cactus.nci.nih.gov}}
The molecules in the database contain up to 104 vertices and 22 vertices on an average.
The weight function $\wIso$ was set to 1 for each pair of vertices and edges with the same label and $-\infty$ otherwise.
This matches the setting in~\cite{Schietgat2013}.

To answer {\bf(H3)} we compared the algorithms on randomly generated connected outerplanar graphs.
Our graph generator takes several parameters as input.
With them we evaluated three different properties: the graph size, the average ratio $|E|/|V|$ of edges to vertices, and the average block size.
For any outerplanar graphs the ratio of edges to vertices is less than 2.
While evaluating the effect of one property, we preserved the other two.
This procedure allows to verify whether our theoretical findings are consistent with the running times observed in practice.
We set the weight function $\wIso$ to 1 for each pair of vertices and edges, which corresponds to uniformly labeled graphs.

\begin{figure}[t]
	\centering
 	\includegraphics[width=.99\textwidth]{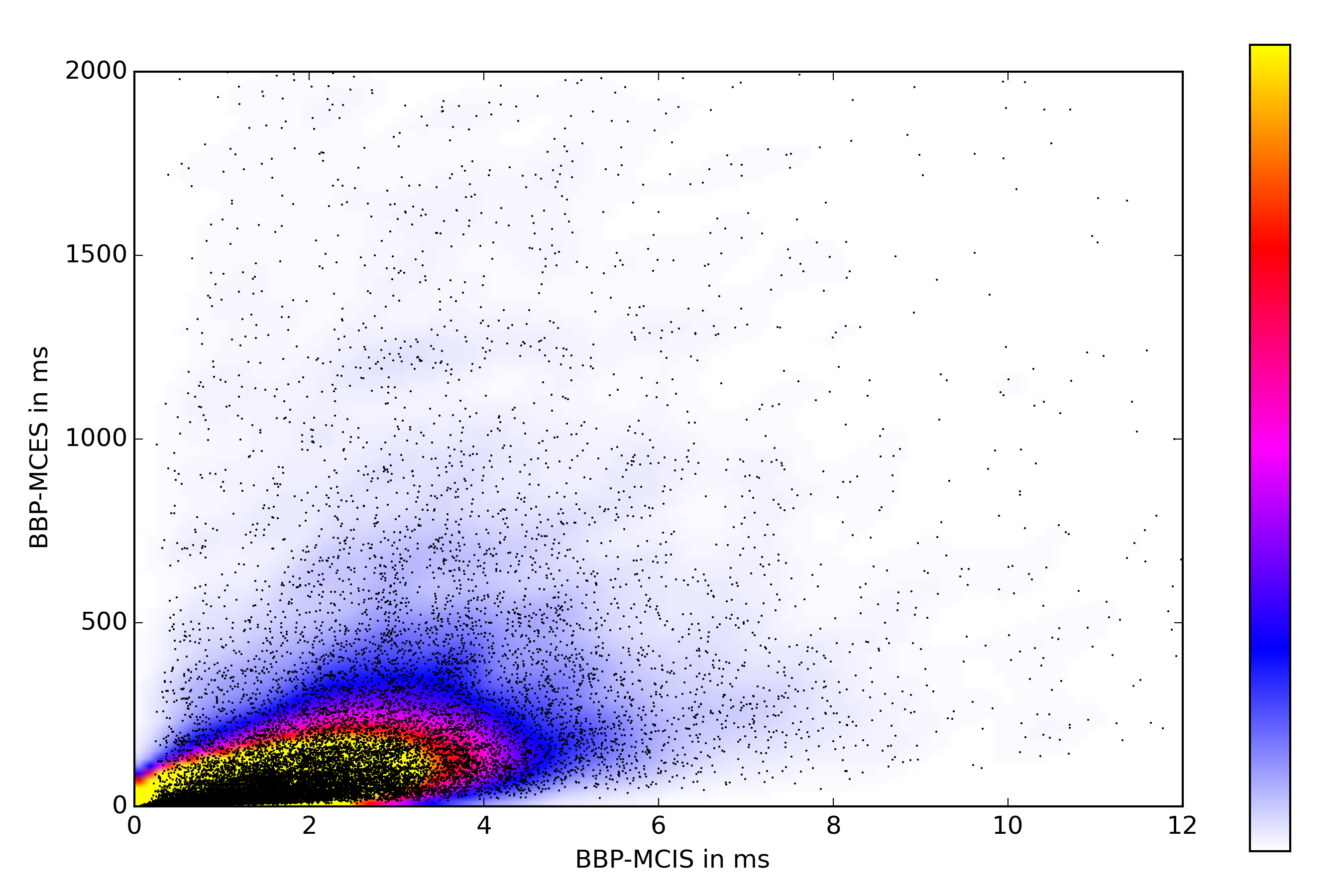}
	\caption{Running times in ms for 28\,399 \BBPMCS computations.
	Each black dot represents a \BBPMCS computation on two randomly chosen outerplanar molecular graphs.
	It directly compares the running time of our algorithm (MCIS, x-axis) and the implementation from~\cite{Schietgat2013} (MCES, y-axis).
	The running times of another 601 \BBPMCS computations did not fit into the borders.
 }
 \label{fig:molecules}
\end{figure}

\begin{table}[t]
\centering
\caption{Upper half: Running times for our implementation (MCIS) and the implementation from~\cite{Schietgat2013} (MCES). Lower half: Relative differences in computation times.}
\begin{tabular}{!{\vrule width.9pt}c!{\vrule width.9pt}r|r|r|r!{\vrule width.9pt}}
\Xhline{.9pt}
Algorithm&Average time&Median time&\ \ 95\% less than&Maximum time\\
\Xhline{.9pt}
MCIS&$1.97$ ms&$1.51$ ms&$5.28$ ms&$40.35$ ms\\
\hline
MCES&$207.08$ ms&$41.43$ ms&$871.48$ ms&$26\,353.68$ ms\\
\Xhline{.9pt}\noalign{\smallskip}\Xhline{.9pt}
Comparison&Average factor&Median factor&Minimum factor&Maximum factor\\
\Xhline{.9pt}
MCES / MCIS&83.8&25.6&1.8&28912.5\\
\Xhline{.9pt}
\end{tabular}
\label{tab:results:comp}
\end{table}

\noindent\textbf{(H1)}
While comparing the weight of the isomorphisms computed by the two algorithms we observed a difference for only 0.40\% of the 29\,000 tested molecule pairs.
This suggests that \BBPMCS yields a valid notion of similarity for outerplanar molecular graphs as it was shown for BBP-MCES~\cite{Schietgat2013}.

\noindent\textbf{(H2)}
Our algorithm computed the solutions on average 84 times faster.
The dots in Fig.~\ref{fig:molecules} represent the computation times of the two algorithms.
The results are summarized in Table~\ref{tab:results:comp}.
Schietgat et al.\,\cite{Schietgat2013} compared their BBP-MCES algorithm to a state-of-the-art algorithm for general MCIS.
Their algorithm had similar computation times on small graphs and was much faster on large graphs.
The maximum time of the general MCIS algorithm was more than 24 hours.
In contrast, our computation time never exceeded 41 ms. This clearly indicates 
that our algorithm is orders of magnitude faster than the general approach.

\begin{table}[t]
\centering
\caption{Average time $\pm$ SD over 100 \BBPMCS computations on random outerplanar graphs, varying one property (graph size, ratio of edges to vertices, block size BS). Note the units of measurement; timeout---total time exceeds 3 days.}
\begin{tabular}{!{\vrule width.9pt}c!{\vrule width.9pt}r|r|r|r|r!{\vrule width.9pt}}
\Xhline{.9pt}
Size&10&20&40&80&160\\
\Xhline{.9pt}
MCIS&$0.7\pm0.3$ ms& $2.3\pm0.8$ ms&$8.2\pm1.6$ ms &$33.5\pm3.6$ ms & $133.2\pm10.1$ ms  \\
\hline
MCES& $207\pm118$ ms&$3.4\pm6.0$ s &$38.6\pm90.6$ s &$234.2\pm420.9$ s & timeout \\
\Xhline{.9pt}\noalign{\smallskip}\Xhline{.9pt}
$|E|/|V|$&0.98&1.10&1.24&1.46&1.78\\
\Xhline{.9pt}
MCIS&  $3.8\pm0.3$ ms &  $4.0\pm1.1$ ms &$8.2\pm1.6$ ms  &$30.8\pm4.0$ ms  &$110.3\pm11.6$ ms \\
\hline
MCES& $223\pm16$ ms& $2.2\pm2.6$ s&$38.6\pm90.6$ s  &$111.0\pm213.8$ s & $216.1\pm288.3$ s\\
\Xhline{.9pt}\noalign{\smallskip}\Xhline{.9pt}
BS&3&5&10&20&40\\
\Xhline{.9pt}
MCIS&  $27\pm6.4$ ms&$13.3\pm2.4$ ms & $8.4\pm1.7$ ms &$5.5\pm1.4$ ms  & $4.5\pm0.9$ ms \\
\hline
MCES& $132\pm14$ ms&$689\pm548$ ms & $83.7\pm118.7$ s & $30.4 \pm 27.8$ min& timeout\\
\Xhline{.9pt}
\end{tabular}
\label{tab:results:random}
\end{table}

\noindent\textbf{(H3)}
We first varied the size of the input graphs, while preserving an average ratio of edges to vertices of $1.24$ and an average block size of $8$. Based on Theorem~\ref{th:runtime:bbpmcs} we expected the average time to increase by a factor of a bit more than 4, if we double the size. The results in Table~\ref{tab:results:random} closely match this expectation.

Next, we evaluated different ratios of edges to vertices.
The graph size was set to 40 and the average block size to 8.
A higher ratio results in a higher number of faces in the blocks and consequently 
affects the time required by Alg.~\ref{alg:2mcs:outer}.
In particular, the table size and, thus, the running time is expected to show a quadratic growth. 
The increase in running time exceeds our expectation. This might be explained by
the increasing size of the data structure used to represent the faces of the blocks.
  
Finally, we evaluated different average block sizes.
The graph size was set to 40 and the average ratio of edges to vertices to 1.24.
Higher block sizes mean less MWMs to compute, which are the most costly part in 
the \BBPMCS computation. Therefore we expected the running time to decrease.
The results shown in Table~\ref{tab:results:random} support this.

\section{Conclusion}
We have developed an algorithm, which computes a well-defined, chemical 
meaningful largest common substructure of outerplanar molecular graphs in a 
fraction of a second.
Hence, our method makes the graph-based comparison in large molecular datasets 
possible.
As future work, we would like to extend our approach to more general graph 
classes with a focus on efficiency in practice.

\bibliographystyle{splncs03}
\bibliography{CommonSubtreeEnumeration}

\newpage
\appendix
\section{Pseudocode}
\SetKwFunction{SS}{SetSx}
\SetKwFunction{MCS}{BBP-MCIS}
\SetKwFunction{BC}{BC}
\SetKwFunction{B}{B}
\SetKwFunction{C}{C}
\SetKwFunction{MCC}{CC}
\SetKwFunction{CASEONE}{BBP-Edge}
\SetKwFunction{CASETWO}{BBP-SingleVertex}
\SetKwFunction{CASEBOTH}{BBP-Edge/SingleVertex}
\SetKwFunction{ENUM}{Enum-E}
\SetKwFunction{ENUMT}{Enum-SV}
\SetKwFunction{ENUMB}{Enum-E/SV}

\algorithmstyle{ruled}
\begin{algorithm}
  \caption{\BBPMCS of outerplanar graphs}
  \label{alg:bbpmcs}
  \Input{Connected outerplanar graphs $G$ and $H$,\\
  \ weight function $\wIso:(V_G\times V_H)\cup (E_G\times E_H)\to \mathbb{R}^{\geq 0}\cup \{-\infty\}$.}
  \Output{A \BBPMCS of $G$ and $H$.}
  \Data{BC-trees $\BC^G$ and $\BC^H$ with node sets $B^G, C^G, B^H, C^H$.}
  Select an arbitrary B-node (block or bridge) $b\in B^G$.\;
  $\phi\leftarrow \SS(b,\emptyset)$  \note*[r]{Initial recursion call.}
  [Output $\phi$] 
\end{algorithm}

\algorithmstyle{boxed}
\begin{algorithm}
 \Procedure{$\SS(b,X)$}
 {
  \label{alg:bbpmcs:ss}
  \Input{B-node $b\in B^G$, excluded vertices $X\subseteq V(G)$.}
  \Output{Isomorphism of maximum weight on $\MCC(V_G\setminus X,V_b)$.}
  $\varphi\gets (\emptyset\to\emptyset)$ \note*[r]{Initialize as empty mapping}
  \ForAll{B-nodes $\bar b\in B^H$, where $b$ and $\bar b$ are both bridges or both blocks} {
	$\phi\leftarrow \CASEONE(b,\bar b,X)$  \note*[r]{$\Part{xy}\subseteq\S_x$, at least one edge is mapped}
	\textbf{if} $\WIso(\phi)>\WIso(\varphi)$ \textbf{then} $\varphi\gets \phi$
  }
  \ForAll{pairs $(v,\bar v)\in (V(b)\setminus X) \times V(H)$} {
	\If{$\wIso(v\bar v)\neq -\infty$}
	{
	  $\phi\leftarrow \CASETWO(b,X,v,\bar v)$  \note*[r]{$\Part{xy}\subseteq\S_x$, single vertex}
	  \textbf{if} $\WIso(\phi)>\WIso(\varphi)$ \textbf{then} $\varphi\gets \phi$
	}
  }
  \ForAll{paths $bcb'$ in $\BC^G$, where $c\notin X$} {
    $\phi\gets \SS(b',V(b))$ \note*[r]{No vertex of $V(b)$ is mapped, recursion}
	\textbf{if} $\WIso(\phi)>\WIso(\varphi)$ \textbf{then} $\varphi\gets \phi$
  }
  \Return $\varphi$
 }
\end{algorithm}

\begin{algorithm}
 \Procedure{$\CASEONE(b,\bar b,X,v,\bar v)$}
 {
  \Input{B-nodes $b\in B^G, \bar b\in B^H, X\subseteq V(G)$, mapping $v\mapsto \bar v$ (optional).}
  \Output{Maximum isomorphism $\varphi$, where \emph{at least one edge} of $b$ is mapped to $\bar b$; restricted to $\MCC(V_G\setminus X,V_b)$ and $\varphi(v)=\bar v$ (if given).}
  \If {exactly one of $b,\bar b$ is a block}
  {
    \Return $\emptyset\to \emptyset$ \note*[r]{not block and bridge preserving}
  }
  \ForAll{valid isomorphisms $\varphi:V(b)\to V(\bar b)$\label{alg:bbpmcs:valid}}{
    \ForAll{pairs $(c,\bar c)\neq (v,\bar v)$ of cutvertices mapped by $\varphi$\label{alg:bbpmcs:pairs}}{
    
      \ForAll{pairs $(b_i,\bar b_j)\in \B^G\times \B^H$ where $bcb_i$ and $\bar b\bar c\bar b_j$ are paths}{
        $w(b_i,\bar b_j)\gets \WIso(\CASEONE(b_i,\bar b_j,X,c,\bar c))-\wIso(c,\bar c)$
      }
      Compute \MWM $M$ on bipartite graph with edge weights $w(b_i,\bar b_j)$\;
      \ForAll{edges $b_i\bar b_j\in M$}{
        Extend $\varphi$ by $\CASEONE(b_i,\bar b_j,X,c,\bar c)$
      }
    }
	\textbf{if} $\WIso(\phi)>\WIso(\varphi)$ \textbf{then} $\varphi\gets \phi$
  }
  \Return $\varphi$
 }
\end{algorithm}

\begin{algorithm}
 \Procedure{$\CASETWO(b,X,v,\bar v)$}
 {
  \Input{B-node $b\in B^G$, excluded vertices $X\subseteq V(G)$, mapping $v\mapsto \bar v$.}
  \Output{Maximum isomorphism $\varphi$, where \emph{a single vertex} of $V(b)$ is mapped to $V(\bar b)$; restricted to $\MCC(V_G\setminus X,V_b)$, $\varphi(v)=\bar v$.}
  $\varphi\gets (v\mapsto \bar v$)\;
  \If{$v\notin C^G$}  
  {
    \Return $\varphi$ \note*[r]{No expansion possible}
  }
  \eIf{$\bar v\in \C^H$}
  {
    \ForAll{pairs $(b_i,\bar b_j)\in \B^G\times \B^H$, where $bvb_i$ is a path and $\bar v\in \bar b_j$}{
      $w(b_i,\bar b_j)\gets \WIso(\CASEONE(b_i,\bar b_j,X,v,\bar v))-\wIso(v,\bar v)$
    }
    Compute \MWM $M$ on bipartite graph with edge weights $w(b_i,\bar b_j)$\;
    \ForAll{edges $b_i\bar b_j\in M$}{
      Extend $\varphi$ by $\CASEONE(b_i,\bar b_j,X,v,\bar v)$
    }
  }
  {
    $\bar b\gets$ the unambiguous node of the set $\B^H$ that contains the vertex $\bar v$\;
    \ForAll{$b_i\in \B^G$, where $bvb_i$ is a path}{
      $\phi\gets \CASEONE(b_i,\bar b,X,v,\bar v)$\;
      \textbf{if} $\WIso(\phi)>\WIso(\varphi)$ \textbf{then} $\varphi\gets \phi$
    }
  }
  \Return $\varphi$
 }
\end{algorithm}

\end{document}